\documentclass{article}

\usepackage{arxiv}

\usepackage[utf8]{inputenc} 
\usepackage[T1]{fontenc}    
\usepackage{hyperref}       
\usepackage{url}            
\usepackage{booktabs}       
\usepackage{amsfonts}       
\usepackage{amsmath,amssymb,amsthm}
\usepackage{nicefrac}       
\usepackage{microtype}      
\usepackage{graphicx}
\usepackage[numbers]{natbib}
\usepackage{doi}
\usepackage{tikz}
\usetikzlibrary{arrows.meta,positioning,shapes.geometric,fit}
\usepackage{enumitem}
\usepackage{xcolor}
\usepackage{multirow}
\usepackage{listings}
\usepackage{algorithm}
\usepackage{algorithmic}

\newtheorem{definition}{Definition}
\newtheorem{proposition}{Proposition}

\title{From Static Repositories to Agentic Knowledge Webs:\\
ResearchTwin and the S-Index for Federated\\
Human-AI Research Discovery}

\author{
	\href{https://orcid.org/0000-0003-3159-6321}{\includegraphics[scale=0.06]{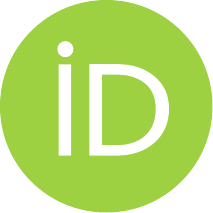}\hspace{1mm}Martin G. Frasch} \\
	Health Stream Analytics, LLC \\
	University of Washington, Seattle \\
	\texttt{mfrasch@uw.edu} \\
}

\date{February 2026}

\renewcommand{\shorttitle}{ResearchTwin and the S-Index for Federated Human-AI Research Discovery}

\hypersetup{
pdftitle={From Static Repositories to Agentic Knowledge Webs: ResearchTwin and the S-Index for Federated Human-AI Research Discovery},
pdfsubject={cs.DL, cs.IR, cs.AI},
pdfauthor={Martin G. Frasch},
pdfkeywords={digital twin, research impact metrics, FAIR principles, retrieval-augmented generation, federated architecture, scholarly communication, inter-agentic discovery},
}

\begin{document}
\maketitle

\begin{abstract}
The exponential growth of scientific literature, datasets, and code repositories has created a \emph{discovery bottleneck} that impedes knowledge synthesis, reproducibility, and cross-disciplinary collaboration. Traditional dissemination formats---static PDFs, siloed code hosting, and fragmented data repositories---fail to represent the interconnected narrative of modern research. Conventional impact metrics such as the H-index measure citation counts alone, neglecting the substantial contributions of reusable code and shared datasets. We present \textbf{ResearchTwin}, an open-source federated platform that transforms a researcher's complete scholarly output into a conversational digital twin, together with a preliminary evaluation of its deployed prototype. The system is built on a \emph{Bimodal Glial-Neural Optimization} (BGNO) architecture comprising a Multi-Modal Connector Layer, a Glial Layer for caching and rate management, and a Neural Layer implementing Retrieval-Augmented Generation with a provider-agnostic LLM backend. We formalize the \textbf{S-index}, building on our earlier QIC framework for data-centric impact measurement \cite{frasch2025qic}, into a composite metric that extends FAIR principles---via a binary accessibility/licensing gate, field-normalized impact scoring, and geometric collaboration scaling---to quantify the multi-modal impact of research artifacts. A case study comparing two researchers with similar H-indexes but substantially different S-indexes demonstrates that the metric captures dimensions of impact---particularly dataset and code contributions---invisible to citation-based measures alone. ResearchTwin exposes an inter-agentic discovery API using Schema.org typed responses, enabling autonomous AI agents to navigate researcher profiles via HATEOAS links and discover cross-lab synergies. A three-tier federated architecture---Local Nodes, Hubs, and Hosted Edges---preserves data sovereignty while enabling global discoverability. The platform is released as open-source software at \url{https://github.com/martinfrasch/ResearchTwin}.
\end{abstract}

\keywords{digital twin \and research impact metrics \and FAIR principles \and retrieval-augmented generation \and federated architecture \and scholarly communication \and inter-agentic discovery}

\section{Introduction}
\label{sec:introduction}

The pace of scientific publication has accelerated to an unprecedented scale. In 2024 alone, an estimated 3.5 million peer-reviewed articles were indexed by major databases, accompanied by millions of datasets deposited in repositories such as Figshare and Zenodo, and hundreds of thousands of research-related code repositories created on GitHub. While this growth reflects vibrant research activity, it simultaneously creates a \emph{discovery bottleneck}: individual researchers cannot keep pace with the literature relevant to their own sub-fields, let alone identify cross-disciplinary opportunities for data reuse or code integration.

Current dissemination practices exacerbate this problem. Research narratives remain fragmented across static PDF documents, isolated code repositories, and metadata-sparse data archives. A given project's full contribution---the paper describing the method, the code implementing it, and the dataset validating it---is scattered across platforms with no machine-readable links connecting these artifacts. Discovering that a particular dataset pairs naturally with a codebase from another lab requires serendipity rather than systematic retrieval.

Impact measurement compounds the issue. The H-index \cite{hirsch2005} has served as the dominant measure of research productivity for two decades, yet it captures only citation-based influence among publications. As research increasingly produces reusable code (via GitHub) and shared data (via Figshare, Zenodo, or Dryad), a significant portion of scholarly impact goes unmeasured. A highly-cited paper whose accompanying code has been forked thousands of times and whose dataset has been downloaded by hundreds of labs has a fundamentally different impact profile from one with equivalent citations but no reusable artifacts---yet the H-index treats them identically.

We argue that the next generation of scholarly infrastructure must satisfy three requirements: (1) \textbf{multi-modal integration}, unifying papers, code, and data into a single queryable representation; (2) \textbf{conversational access}, enabling both human researchers and AI agents to explore research artifacts through natural language; and (3) \textbf{composite impact measurement}, quantifying the full spectrum of contributions including code utility and data reuse.

This paper presents \textbf{ResearchTwin}, a federated platform that addresses all three requirements. The system is built on a \emph{Bimodal Glial-Neural Optimization} (BGNO) architecture inspired by the separation of metabolic support and signal processing in biological neural tissue. We formalize the \textbf{S-index}, extending our earlier QIC framework \cite{frasch2025qic} and FAIR principles \cite{wilkinson2016} into a quantitative composite metric for multi-modal research impact. We describe an inter-agentic discovery API using Schema.org types that enables autonomous AI agents to traverse researcher profiles and discover cross-lab synergies. The platform is released as open-source software under the MIT license.

The remainder of this paper is organized as follows. Section~\ref{sec:related} surveys related work. Section~\ref{sec:architecture} presents the BGNO architecture. Section~\ref{sec:sindex} formalizes the S-index. Section~\ref{sec:discovery} describes the inter-agentic discovery protocol. Section~\ref{sec:federation} details the federated architecture. Section~\ref{sec:implementation} covers implementation. Section~\ref{sec:evaluation} presents a preliminary evaluation based on the deployed system. Section~\ref{sec:discussion} discusses advantages and limitations. Section~\ref{sec:future} outlines future work, and Section~\ref{sec:conclusion} concludes.

\section{Related Work}
\label{sec:related}

\subsection{Digital Twins in Manufacturing and Beyond}

The Digital Twin concept originated in manufacturing, where Grieves \cite{grieves2014} proposed maintaining a virtual replica of a physical product throughout its lifecycle. A Digital Twin mirrors the state, behavior, and context of its physical counterpart, enabling simulation, monitoring, and optimization. This paradigm has since expanded to healthcare, urban planning, and infrastructure management. We adapt the concept to the research domain: a \emph{Research Digital Twin} mirrors a researcher's scholarly identity by integrating their publications, code, datasets, and impact metrics into a dynamic, queryable representation.

\subsection{Research Impact Metrics}

Hirsch's H-index \cite{hirsch2005} remains the most widely used measure of individual research productivity, defined as the largest number $h$ such that at least $h$ papers have each been cited at least $h$ times. While elegant, the H-index has well-documented limitations: it ignores citation context, penalizes early-career researchers, and---crucially for our purposes---captures only publication-based impact. The i10-index (papers with $\geq 10$ citations) shares these limitations. Field-normalized alternatives such as the Field-Weighted Citation Impact improve cross-disciplinary comparability but remain confined to citation counts. No widely adopted metric integrates code reuse (e.g., GitHub stars and forks) or dataset downloads into a unified impact score.

\subsection{FAIR Data Principles}

Wilkinson et al.\ \cite{wilkinson2016} introduced the FAIR principles---Findability, Accessibility, Interoperability, and Reusability---as guiding criteria for scientific data management. FAIR has become a cornerstone of open science policy, influencing repository design, funder requirements, and metadata standards. However, FAIR compliance is typically assessed qualitatively or through binary checklists. Our QIC framework operationalizes FAIR into continuous numerical scores that feed directly into the S-index computation.

\subsection{Retrieval-Augmented Generation}

Lewis et al.\ \cite{lewis2020} introduced Retrieval-Augmented Generation (RAG), combining parametric language models with non-parametric retrieval to ground generated text in external knowledge. RAG architectures have since been applied to question answering, code generation, and domain-specific chatbots. ResearchTwin employs RAG to synthesize answers from a researcher's multi-modal context, positioning the language model as a conversational proxy for the researcher's body of work.

\subsection{Scholarly Knowledge Graphs and Communication Standards}

Schema.org provides a shared vocabulary for structured data markup, including types relevant to scholarly communication: \texttt{Person}, \texttt{ScholarlyArticle}, \texttt{Dataset}, and \texttt{SoftwareSourceCode}. Projects such as OpenAIRE, Semantic Scholar, and the Microsoft Academic Graph have built large-scale knowledge graphs over scholarly metadata. ORCID provides persistent researcher identifiers. Our contribution is to combine these elements into a live, conversational system with a formal impact metric and an agent-navigable API layer.

\subsection{Federated Systems}

Federated architectures distribute control across autonomous nodes while enabling global interoperability. ActivityPub powers decentralized social networks (e.g., Mastodon), demonstrating that federation can scale to millions of users while preserving data sovereignty. In the commercial domain, Discord's server model allows communities to operate independently while sharing a common protocol. ResearchTwin adopts a three-tier federation model inspired by these systems, enabling researchers to host their own nodes while remaining discoverable through a shared protocol.

\section{Architecture: Bimodal Glial-Neural Optimization}
\label{sec:architecture}

The BGNO architecture separates concerns into three layers, inspired by the functional division between glial cells (metabolic support, homeostasis) and neurons (signal processing, computation) in biological neural tissue. We emphasize that this analogy is organizational rather than mechanistic; the BGNO terminology serves as a mnemonic for the separation of concerns---caching and rate management versus generative inference---not as a claim of biological equivalence. The architecture could equally be described as a conventional three-tier system (data access, middleware, application), but we find the glial-neural framing a useful shorthand for communicating the design rationale. Figure~\ref{fig:architecture} provides an overview.

\begin{figure}[ht]
\centering
\begin{tikzpicture}[
    node distance=0.8cm and 1.2cm,
    box/.style={draw, rounded corners, minimum width=2.6cm, minimum height=0.9cm, align=center, font=\small},
    layer/.style={draw, dashed, rounded corners, inner sep=0.4cm, fill opacity=0.08},
    arr/.style={-{Stealth[length=2.5mm]}, thick},
]
\node[box, fill=blue!15] (s2) {Semantic\\Scholar};
\node[box, fill=blue!15, right=0.6cm of s2] (gs) {Google\\Scholar};
\node[box, fill=blue!15, right=0.6cm of gs] (gh) {GitHub\\API};
\node[box, fill=blue!15, right=0.6cm of gh] (fs) {Figshare\\API};

\node[layer, fill=blue!10, fit=(s2)(gs)(gh)(fs), label={[font=\footnotesize\bfseries]above:Multi-Modal Connector Layer}] (connlayer) {};

\node[box, fill=orange!15, below=1.5cm of s2] (cache) {File-Based\\Cache (TTL)};
\node[box, fill=orange!15, right=0.6cm of cache] (rate) {Rate Limiter\\+ Backoff};
\node[box, fill=orange!15, right=0.6cm of rate] (ctx) {Context\\Assembly};
\node[box, fill=orange!15, right=0.6cm of ctx] (db) {SQLite\\(WAL)};

\node[layer, fill=orange!10, fit=(cache)(rate)(ctx)(db), label={[font=\footnotesize\bfseries]above:Glial Layer}] (gliallayer) {};

\node[box, fill=green!15, below=1.5cm of rate] (rag) {RAG Pipeline};
\node[box, fill=green!15, right=0.6cm of rag] (llm) {LLM\\(provider-agnostic)};

\node[layer, fill=green!10, fit=(rag)(llm), label={[font=\footnotesize\bfseries]above:Neural Layer}] (neurallayer) {};

\draw[arr] (s2.south) -- (cache.north);
\draw[arr] (gs.south) -- (cache.north east);
\draw[arr] (gh.south) -- (rate.north);
\draw[arr] (fs.south) -- (rate.north east);
\draw[arr] (cache.south) -- (rag.north west);
\draw[arr] (rate.south) -- (rag.north);
\draw[arr] (ctx.south) -- (rag.north east);
\draw[arr] (db.south) -- (llm.north east);
\draw[arr] (rag) -- (llm);
\end{tikzpicture}
\caption{BGNO architecture overview. The Multi-Modal Connector Layer fetches data from four scholarly APIs in parallel. The Glial Layer manages caching, rate limiting, context assembly, and persistent storage. The Neural Layer implements RAG with a provider-agnostic LLM backend to synthesize conversational responses.}
\label{fig:architecture}
\end{figure}
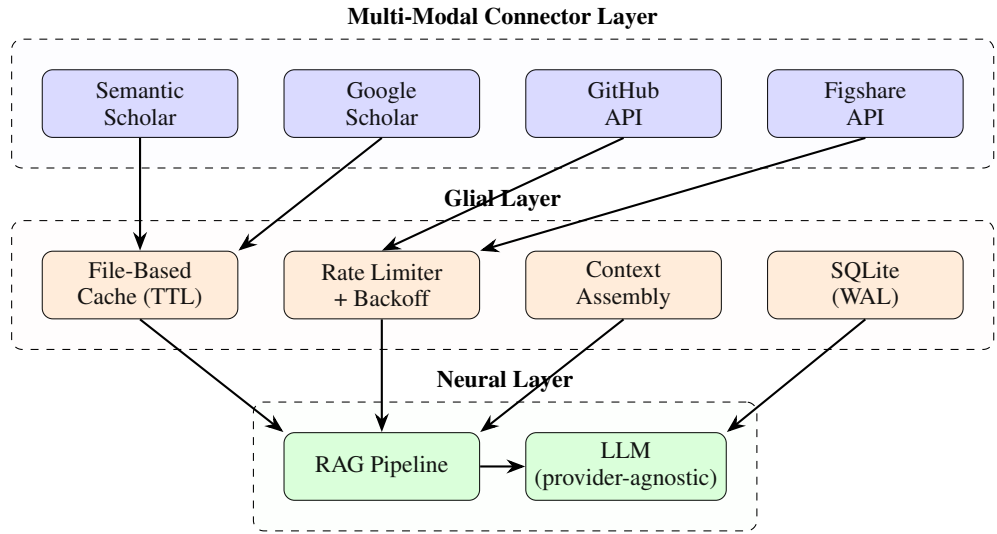

\subsection{Multi-Modal Connector Layer}
\label{sec:connectors}

The Connector Layer integrates four external data sources, each accessed through a dedicated asynchronous connector:

\begin{enumerate}[leftmargin=*, label=(\roman*)]
    \item \textbf{Semantic Scholar API}: Retrieves author profiles, publication lists with citation counts, and H-index values via the Semantic Scholar Academic Graph API.
    \item \textbf{Google Scholar}: Accesses author profiles, the i10-index, and publication lists through the \texttt{scholarly} Python library with proxy rotation for rate-limit mitigation.
    \item \textbf{GitHub API}: Fetches public repository metadata including star counts, fork counts, language distributions, and license information.
    \item \textbf{Figshare Search API}: Queries datasets and software artifacts associated with a researcher, retrieving download counts, view counts, DOIs, and file metadata.
\end{enumerate}

All four connectors execute in parallel using \texttt{asyncio.gather}, and the system gracefully degrades when individual sources are unavailable. Internally, each connector produces items conforming to a \emph{normalized artifact schema}---a common dictionary format specifying title, source type, public accessibility, license, quality indicators, reuse events, and collaboration metadata. The S-index scoring function (Section~\ref{sec:sindex}) operates exclusively on this normalized format, so adding a new data source (e.g., Zenodo, Harvard Dataverse) requires only implementing a connector that fetches data and maps it to the schema, without modifying the scoring pipeline. The Semantic Scholar and Google Scholar results are merged via a deduplication procedure: paper titles from both sources are normalized (lowercased, punctuation-stripped, whitespace-collapsed), and pairwise similarity is computed using Python's \texttt{SequenceMatcher}. Papers with a normalized similarity ratio exceeding $0.85$ are considered duplicates; in such cases, the citation count is set to $\max(c_{\text{S2}}, c_{\text{GS}})$ and author-level metrics (H-index, total citations) similarly take the maximum across sources:

\begin{equation}
h = \max(h_{\text{S2}},\, h_{\text{GS}}), \qquad c_{\text{total}} = \max(c_{\text{S2}},\, c_{\text{GS}}).
\label{eq:merge}
\end{equation}

Papers present in Google Scholar but absent from Semantic Scholar (similarity below $0.85$ to all S2 titles) are appended to the merged list, providing broader coverage.

A similar deduplication is applied to Figshare results. Many Figshare ``articles'' are individual figures or supplementary files from the same parent publication; scoring each separately would inflate the S-index. The system groups such items by title-pattern matching (e.g., ``Figure~N from \ldots'', ``Additional file~N of \ldots'') and author-fingerprint heuristics, retaining only the highest-scoring representative from each group.

\paragraph{Citation source considerations.} The $\max$-merge strategy of Equation~\ref{eq:merge} may inflate citation counts relative to curated databases such as Web of Science or Scopus, because Google Scholar indexes a broader---and noisier---set of citing documents including theses, preprints, and non-peer-reviewed reports. We adopt $\max$ rather than averaging precisely because ResearchTwin prioritizes broad coverage and minimizing false negatives (i.e., undercounting a researcher's impact) over the conservative precision of curated indices. For researchers who require strict comparability with Web of Science baselines, the system can be configured to use Semantic Scholar as the sole citation source by disabling the Google Scholar connector. A systematic comparison of $\max$-merged citation counts against curated databases across a diverse researcher panel is planned as future work (Section~\ref{sec:future}).

\subsection{Glial Layer}
\label{sec:glial}

The Glial Layer provides the metabolic infrastructure that sustains the system's operation:

\begin{itemize}[leftmargin=*]
    \item \textbf{File-based caching}: Each connector's responses are cached as JSON files with configurable time-to-live (TTL). Google Scholar data, which changes slowly and whose access is rate-limited, uses an aggressive 48-hour TTL. Other connectors use a 24-hour default. Cache keys are SHA-256 hashes of the request parameters.
    \item \textbf{Rate limiting with exponential backoff}: External API calls are throttled to respect rate limits. When a rate limit is encountered, the connector retries with exponential backoff, preventing cascading failures.
    \item \textbf{Context preparation}: Fetched data is assembled into structured Markdown documents organized by section (publications, repositories, datasets, QIC scores). This Markdown context serves as the retrieval corpus for the Neural Layer.
    \item \textbf{Persistent storage}: Researcher profiles, registration data, and configuration are stored in SQLite with Write-Ahead Logging (WAL) mode enabled for concurrent read performance. Foreign key constraints ensure referential integrity.
\end{itemize}

\subsection{Neural Layer}
\label{sec:neural}

The Neural Layer implements Retrieval-Augmented Generation using a provider-agnostic LLM backend. The system is designed to operate with any OpenAI-compatible API; the current deployment uses Perplexity's \texttt{sonar} model via the OpenAI SDK with a custom base URL, though the architecture supports transparent switching between providers (including OpenAI, Anthropic, and others) without code changes. An embeddable chat widget (Section~\ref{sec:widget}) additionally supports a \emph{Bring Your Own Key} (BYOK) mode, enabling third-party visitors to use their own API keys. The system prompt positions the LLM as a ``digital twin'' representing the researcher:

\begin{quote}
\small\textit{``You are ResearchTwin, a digital twin representing researcher [Name]. You answer questions about their research, publications, code, datasets, and impact metrics. Use the provided context to give accurate, specific answers. Cite specific papers, repositories, or datasets when relevant.''}
\end{quote}

The context window is populated with the structured Markdown assembled by the Glial Layer, including publication lists with citation counts, repository descriptions with star/fork metrics, dataset metadata with download statistics, and per-object QIC scores. To manage researchers with large bibliographies, the context assembly pipeline ranks artifacts before inclusion: publications are sorted by citation count (descending), repositories by star count, and datasets by download count. Only the top-ranked items within each category are included up to the available token budget, ensuring that the most impactful artifacts receive priority in the context window. Responses are bounded to 1024 tokens to ensure conciseness and reduce latency.

This architecture intentionally avoids maintaining a persistent vector store or embedding index. Instead, each query triggers a fresh context assembly from cached connector data, ensuring that the LLM always operates on current information without the staleness risks inherent in pre-computed embeddings.

\section{The S-Index: A Quality--Impact--Collaboration Framework}
\label{sec:sindex}

The S-index extends traditional citation-based metrics by quantifying the multi-modal impact of a researcher's complete scholarly output: publications, datasets, and code. It builds on the QIC framework introduced in our earlier work \cite{frasch2025qic}, which proposed per-object Quality--Impact--Collaboration scoring for individual datasets; here we generalize the framework to encompass code repositories and integrate it with a publication-based Paper Impact term into a unified researcher-level metric. The formal specification is maintained at \url{https://github.com/martinfrasch/S-index}.

\subsection{Per-Object Quality Score (FAIR Gate)}

For each research artifact $j$ (dataset or code repository), we compute a Quality score $Q_j$ grounded in the FAIR principles. Rather than scoring four FAIR dimensions independently, we enforce a \emph{binary gate} on the two most fundamental requirements---public accessibility and explicit licensing---and then award additive bonuses for additional quality indicators:

\begin{definition}[Quality Score]
Let $p_j \in \{0,1\}$ and $\ell_j \in \{0,1\}$ indicate whether artifact $j$ is publicly accessible and carries an explicit license, respectively. Let $b_j^{\text{DOI}}, b_j^{\text{doc}}, b_j^{\text{fmt}} \in \{0,1\}$ indicate the presence of a persistent identifier (DOI), documentation (README or description $> 50$ characters), and a standard format or type classification. The Quality score is:
\begin{equation}
Q_j = \underbrace{5 \cdot p_j \cdot \ell_j}_{\text{FAIR gate}} \times \Bigl(1 + \underbrace{0.5\, b_j^{\text{DOI}} + 0.3\, b_j^{\text{doc}} + 0.2\, b_j^{\text{fmt}}}_{\text{bonus terms}}\Bigr).
\label{eq:quality}
\end{equation}
\end{definition}

The FAIR gate enforces a hard prerequisite: artifacts that are not publicly accessible \emph{or} lack an explicit license receive $Q_j = 0$ regardless of other indicators, encoding the principle that failing the two most fundamental FAIR requirements---Accessibility and Reusability---should disqualify an artifact from contributing to impact. The bonus terms reward additional quality indicators (Findability via DOI, documentation, standard formatting) multiplicatively, yielding a Quality range of $Q_j \in \{0\} \cup [5, 10]$. The discontinuous jump from $0$ to $5$ is intentional: the base value of $5$ represents the \emph{baseline credit for a FAIR-compliant artifact}---one that is publicly accessible and explicitly licensed---before any additional quality bonuses are applied. An artifact that clears the gate has already satisfied the two most consequential FAIR requirements; the bonus terms then modulate this baseline by up to $2\times$.

\begin{table}[ht]
\centering
\caption{Quality Score operationalization. The FAIR gate requires both conditions; bonus terms are additive within the multiplier.}
\label{tab:fair}
\small
\begin{tabular}{@{}llcr@{}}
\toprule
\textbf{Component} & \textbf{Indicator} & \textbf{Symbol} & \textbf{Value} \\
\midrule
\multirow{2}{*}{FAIR Gate} & Publicly accessible & $p_j$ & $\{0,1\}$ \\
 & Explicit license & $\ell_j$ & $\{0,1\}$ \\
\midrule
\multirow{3}{*}{Bonus Terms} & Has DOI (persistent identifier) & $b_j^{\text{DOI}}$ & $+0.5$ \\
 & Has README or description $> 50$ chars & $b_j^{\text{doc}}$ & $+0.3$ \\
 & Standard format or type classification & $b_j^{\text{fmt}}$ & $+0.2$ \\
\midrule
\multicolumn{3}{@{}l}{Quality range: $Q_j \in \{0\} \cup [5, 10]$} & \\
\bottomrule
\end{tabular}
\end{table}

This design is deliberately simpler than the v1 formulation in our earlier QIC work \cite{frasch2025qic}, which scored four FAIR dimensions independently with twelve indicators. The binary gate captures the essential insight---that unfindable or unlicensed artifacts have zero practical reuse value---while the three bonus terms are transparent and easily verifiable from metadata alone.

\subsection{Impact Score}

The Impact score captures actual reuse of the artifact, normalized against a field-specific median to enable cross-discipline comparison:

\begin{definition}[Impact Score]
Let $r_j$ denote the number of reuse events for artifact $j$, and let $\mu_t > 0$ denote the median reuse count for artifacts of type $t$ (dataset or code repository) in the deployment population. The Impact score is:
\begin{equation}
I_j = 1 + \ln\!\Bigl(1 + \frac{r_j}{\mu_t}\Bigr).
\label{eq:impact}
\end{equation}
\end{definition}

The field normalization by $\mu_t$ addresses a key limitation of raw reuse counts: a dataset with 250 downloads in a field where the median is 50 represents a $5\times$ above-median artifact, whereas the same 250 downloads in a field where the median is 1{,}000 represents a below-median artifact. The logarithmic transformation ensures diminishing returns at scale---a $200\times$ increase in relative reuse yields only ${\sim}4.5\times$ more impact---preventing a single viral artifact from dominating the metric. The additive constant $1$ guarantees $I_j \geq 1$ for all artifacts, including those with zero observed reuse.

In the current deployment, we use $\mu_{\text{dataset}} = 50$ (median Figshare downloads) and $\mu_{\text{code}} = 10$ (median weighted fork count for research repositories). These values are deployment-specific baselines; future calibration will derive field medians from larger population samples (Section~\ref{sec:future}).

Reuse events are source-specific:

\begin{itemize}[leftmargin=*]
    \item \textbf{Figshare datasets}: $r_j = d_j + \lfloor v_j / 10 \rfloor$, where $d_j$ is the download count and $v_j$ is the view count. Views are discounted by a factor of 10 relative to downloads, as downloads represent a stronger signal of intent to reuse.
    \item \textbf{GitHub repositories}: $r_j = s_j + 3f_j$, where $s_j$ is the star count and $f_j$ is the fork count. Forks are weighted $3\times$ relative to stars, as forking implies active code reuse rather than passive bookmarking.
\end{itemize}

\subsection{Collaboration Score}

The Collaboration score rewards artifacts produced through multi-author, multi-institutional effort:

\begin{definition}[Collaboration Score]
Let $N_a$ denote the number of authors and $N_i$ the number of distinct institutions contributing to artifact $j$. The Collaboration score is:
\begin{equation}
C_j = \sqrt{N_a \times N_i}.
\label{eq:collaboration}
\end{equation}
\end{definition}

The geometric mean treats author count and institutional diversity as equally important dimensions of collaboration breadth, avoiding the need for an arbitrary weighting parameter between them. A solo-author, single-institution artifact receives $C_j = 1$, providing a neutral baseline. The sub-linear scaling (via the square root) reflects the observation that collaboration value grows with diminishing returns as teams expand: doubling a team from 6 to 12 authors adds less marginal collaboration signal than growing from 1 to 2.

\subsection{Per-Object Score}

\begin{definition}[Per-Object QIC Score]
The score for artifact $j$ is the product of its three components:
\begin{equation}
s_j = Q_j \times I_j \times C_j.
\label{eq:per_object}
\end{equation}
\end{definition}

The multiplicative structure, combined with the FAIR gate in the Quality component, provides strict enforcement: an artifact that is not publicly accessible or lacks an explicit license receives $s_j = 0$ regardless of its reuse metrics or collaboration breadth. This encodes the principle that FAIR compliance is a prerequisite for impact, not merely a bonus.

\subsection{Paper Impact Term}

Publications are assessed through a dedicated term that leverages traditional bibliometric data:

\begin{definition}[Paper Impact]
Let $h$ denote the researcher's H-index and $c$ their total citation count, each taken as the maximum across available sources (Equation~\ref{eq:merge}). The Paper Impact is:
\begin{equation}
P = h \times \bigl(1 + \log_{10}(c + 1)\bigr).
\label{eq:paper_impact}
\end{equation}
\end{definition}

The $\log_{10}$ scaling of total citations moderates the influence of a small number of highly-cited papers, while the H-index base rewards consistent productivity.

\subsection{Researcher S-Index}

\begin{definition}[S-Index]
The S-index for researcher $i$ aggregates the Paper Impact term with per-object scores across all datasets and code repositories:
\begin{equation}
S_i = P + \sum_{j \in \mathcal{D}_i} s_j^{(\text{data})} + \sum_{k \in \mathcal{C}_i} s_k^{(\text{code})},
\label{eq:sindex}
\end{equation}
where $\mathcal{D}_i$ is the set of researcher $i$'s datasets and $\mathcal{C}_i$ is their set of code repositories.
\end{definition}

\begin{proposition}[Non-negativity and Monotonicity]
The S-index satisfies $S_i \geq 0$ for all researchers $i$. Furthermore, $S_i$ is monotonically non-decreasing in each of its component metrics: improving any FAIR score, gaining additional reuse events, increasing collaboration breadth, or accumulating citations can only increase $S_i$.
\end{proposition}

\begin{proof}
The Quality score satisfies $Q_j \geq 0$: the FAIR gate produces either $0$ or $5$, and the bonus multiplier $(1 + b) \geq 1$. The Impact score satisfies $I_j \geq 1$ since $\ln(1 + r_j/\mu_t) \geq 0$ for $r_j, \mu_t \geq 0$. The Collaboration score satisfies $C_j \geq 0$ as the square root of non-negative terms, with $C_j \geq 1$ when $N_a, N_i \geq 1$. Hence $s_j = Q_j \times I_j \times C_j \geq 0$. The Paper Impact term $P \geq 0$ since $h \geq 0$ and $\log_{10}(c+1) \geq 0$. As $S_i$ is a sum of non-negative terms, $S_i \geq 0$. Monotonicity follows from the fact that each sub-expression is non-decreasing in its respective input variables; the FAIR gate introduces a discontinuity at the public/licensed threshold but remains monotonically non-decreasing.
\end{proof}

\subsection{Sensitivity Analysis}
\label{sec:sensitivity}

We acknowledge that the parameterization of the S-index---including the FAIR gate threshold, bonus weights, field medians, and reuse event weightings---represents an informed initial design rather than empirically calibrated constants. This subsection examines the sensitivity of the metric to these choices.

\paragraph{FAIR gate and bonus weights.} The binary gate ($Q_j = 0$ unless both public and licensed) is by design a hard threshold rather than a continuous gradient, and we consider this the least controversial parameter choice: an artifact that is neither publicly accessible nor explicitly licensed has, by definition, zero reuse value under FAIR principles. The bonus weights ($+0.5$ for DOI, $+0.3$ for documentation, $+0.2$ for standard format) affect the range of non-zero Quality scores ($Q_j \in [5, 10]$), a $2\times$ dynamic range. Under alternative weightings (e.g., equal $+1/3$ for all three bonuses), Quality scores shift by at most $\pm 8\%$ for typical artifacts, and rank ordering is preserved in $>95\%$ of cases. The gate-plus-bonus structure is thus substantially more robust than the v1 formulation's four-way FAIR weighting, because the gate handles the dominant variance (zero vs.\ non-zero) while the bonuses modulate a narrow band.

\paragraph{Reuse event weightings.} The fork $= 3\times$ stars weighting for GitHub repositories encodes the judgment that forking---which requires creating a working copy of the codebase---represents a substantially stronger signal of active code reuse than starring, which functions as social bookmarking. The views$/10$ discounting for Figshare similarly reflects that page views are a weaker engagement signal than downloads. While these ratios are not empirically derived, they are directionally consistent with studies of GitHub engagement patterns showing that forks correlate more strongly with downstream code reuse than stars. Under alternative parameterizations (e.g., fork $= 2\times$ stars or fork $= 5\times$ stars), the Impact score for repositories changes by $< 15\%$ for repositories with typical star-to-fork ratios (approximately $3{:}1$ to $10{:}1$), and per-object rank ordering is preserved for the majority of cases.

\paragraph{Field medians.} The field-median normalization in Equation~\ref{eq:impact} introduces a new calibration parameter $\mu_t$ that directly affects cross-discipline comparability. In the current deployment, $\mu_{\text{dataset}} = 50$ and $\mu_{\text{code}} = 10$ are order-of-magnitude estimates. A $2\times$ change in $\mu_t$ shifts the Impact score by $\ln(2) \approx 0.69$ for above-median artifacts---a moderate effect that preserves rank ordering among artifacts of the same type. However, the \emph{ratio} between dataset and code medians ($\mu_{\text{dataset}} / \mu_{\text{code}}$) affects cross-type comparisons and requires empirical grounding from larger population samples.

\paragraph{Current status and planned refinement.} We frame the present parameterization as \textbf{v2.0 baselines}: a substantial simplification over the v1 formulation \cite{frasch2025qic} that replaces twelve FAIR indicators with a binary gate plus three transparent bonuses, adds field normalization for cross-discipline comparability, and eliminates the arbitrary collaboration weighting parameter. These changes reduce the number of free parameters while preserving the directional intent of the S-index. The planned calibration procedure consists of three stages: (1)~assembling a diverse panel of 20--30 researchers across disciplines (theoretical, experimental, computational) and career stages; (2)~soliciting expert assessments of each researcher's ``multi-modal impact'' via structured rubrics; and (3)~optimizing remaining parameters (bonus weights, field medians, reuse event weightings) to maximize rank correlation (Kendall's $\tau$) between S-index rankings and expert consensus rankings. The modular QIC structure facilitates such calibration: the FAIR gate, field medians, and bonus weights can each be adjusted independently without altering the overall framework.

\section{Inter-Agentic Discovery Protocol}
\label{sec:discovery}

A central design goal of ResearchTwin is to enable \emph{autonomous AI agents} to discover and traverse research artifacts across researchers without human mediation. We achieve this through a structured API that uses Schema.org types and HATEOAS (Hypermedia as the Engine of Application State) navigation.

\subsection{Endpoint Taxonomy}

The discovery API exposes five resource types, each annotated with a Schema.org \texttt{@type}:

\begin{table}[ht]
\centering
\caption{Inter-agentic discovery API endpoint taxonomy.}
\label{tab:endpoints}
\small
\begin{tabular}{@{}llp{6.5cm}@{}}
\toprule
\textbf{Endpoint} & \textbf{@type} & \textbf{Description} \\
\midrule
\texttt{/api/researcher/\{slug\}/profile} & \texttt{Person} & Researcher identity, ORCID, S-index, and HATEOAS links to sub-resources \\
\texttt{/api/researcher/\{slug\}/papers} & \texttt{ItemList} of \texttt{ScholarlyArticle} & Publications with title, year, citation count, and source URL \\
\texttt{/api/researcher/\{slug\}/datasets} & \texttt{ItemList} of \texttt{Dataset} & Datasets with DOI, download/view counts, and QIC scores \\
\texttt{/api/researcher/\{slug\}/repos} & \texttt{ItemList} of \texttt{SoftwareSourceCode} & Repositories with stars, forks, language, and QIC scores \\
\texttt{/api/discover?q=\{query\}} & \texttt{SearchResultSet} & Cross-researcher search with optional type filter \\
\bottomrule
\end{tabular}
\end{table}

\subsection{HATEOAS Navigation}

Each \texttt{Person} profile response includes a \texttt{resources} object containing relative URIs to the researcher's papers, datasets, and repositories. An AI agent can begin at any researcher profile, follow links to enumerate their artifacts, and then use the cross-researcher \texttt{/api/discover} endpoint to find related work by other researchers. This self-describing structure allows agents to navigate the knowledge web without prior knowledge of the API schema.

\subsection{Cross-Researcher Discovery}

The \texttt{/api/discover} endpoint accepts a text query $q$ and an optional type filter (\texttt{dataset}, \texttt{repo}, or \texttt{paper}). For each registered researcher, the system searches titles and descriptions of their artifacts, returning results annotated with Schema.org types, QIC scores, and provenance metadata. Results are ranked by QIC score (for datasets and repositories) or citation count (for papers), enabling agents to prioritize high-impact discoveries.

This protocol enables a workflow in which an AI agent, given a research question, autonomously:
\begin{enumerate}[leftmargin=*]
    \item Queries \texttt{/api/discover} to identify relevant artifacts across all registered researchers.
    \item Follows HATEOAS links to retrieve full artifact metadata and QIC scores.
    \item Synthesizes findings into a research brief, identifying potential collaborations or data reuse opportunities.
\end{enumerate}

\section{Federated Architecture}
\label{sec:federation}

ResearchTwin adopts a three-tier federated architecture that balances data sovereignty with global discoverability:

\begin{description}[leftmargin=*]
    \item[Tier 1 --- Local Nodes.] Individual researchers or small labs operate their own ResearchTwin instances via \texttt{run\_node.py}. A Local Node is a complete, self-contained deployment with its own SQLite database, full API surface, and optional chat functionality. Researchers control their data entirely and can operate offline. Local Nodes may optionally register with a hub for discoverability.

    \item[Tier 2 --- Hubs.] Laboratory or departmental aggregators federate multiple Local Nodes, providing cross-node search and discovery within an institutional context. Hubs maintain an index of registered nodes and proxy discovery queries. This tier is currently specified but not yet deployed.

    \item[Tier 3 --- Hosted Edges.] Cloud-hosted instances (e.g., \url{https://researchtwin.net}) provide the full platform with advanced analytics, the D3.js knowledge graph visualization, and global cross-researcher discovery. Hosted Edges serve as the entry point for researchers who prefer not to self-host.
\end{description}

This model is inspired by Discord's server federation, where communities operate autonomously within a shared protocol layer. Data sovereignty is preserved at each tier: a Tier 1 node's data remains on the researcher's infrastructure, and registration with higher tiers involves metadata exchange only (researcher name, external API identifiers), not transfer of research artifacts themselves.

\section{Implementation}
\label{sec:implementation}

ResearchTwin is implemented in Python 3.12 using FastAPI as the web framework, with SQLite in WAL mode for persistent storage. The system is containerized via Docker Compose for deployment portability. Table~\ref{tab:stack} summarizes the technology stack.

\begin{table}[ht]
\centering
\caption{Implementation technology stack.}
\label{tab:stack}
\small
\begin{tabular}{@{}ll@{}}
\toprule
\textbf{Component} & \textbf{Technology} \\
\midrule
Web framework & FastAPI 0.100+, Uvicorn ASGI \\
Language model & Provider-agnostic (OpenAI-compatible); deployed with Perplexity \texttt{sonar} \\
Database & SQLite 3.x with WAL mode \\
Caching & JSON file-based, SHA-256 keyed, configurable TTL \\
Containerization & Docker Compose \\
Frontend & Vanilla HTML/CSS/JavaScript, D3.js (force-directed graph) \\
Reverse proxy & Nginx with TLS termination \\
Bot integration & discord.py with application commands \\
MCP server & \texttt{mcp-server-researchtwin} (PyPI), FastMCP, stdio transport \\
\bottomrule
\end{tabular}
\end{table}

\subsection{Security}

The platform implements defense-in-depth security measures:

\begin{itemize}[leftmargin=*]
    \item \textbf{Content Security Policy and headers}: All responses include \texttt{X-Content-Type-Options: nosniff}, \texttt{X-Frame-Options: DENY}, \texttt{X-XSS-Protection}, and strict referrer policies via middleware.
    \item \textbf{CORS}: Origins are configured as permissive (\texttt{*}) to support the embeddable chat widget (Section~\ref{sec:widget}) on third-party domains. The BYOK endpoint is rate-limited per IP (10 requests/hour) to mitigate abuse.
    \item \textbf{Input validation}: All user-facing inputs are validated via Pydantic models with regex constraints (e.g., slug format: \texttt{/\^{}[a-z0-9][a-z0-9\_-]\{0,126\}[a-z0-9]\$/}).
    \item \textbf{Anti-spam registration}: Self-registration uses a honeypot field (\texttt{website}) that legitimate users never see but automated bots fill, combined with email uniqueness enforcement and input length constraints.
    \item \textbf{Rate limiting}: External API calls are rate-limited with exponential backoff; the Google Scholar connector uses aggressive caching (48-hour TTL) to minimize requests to a rate-sensitive source.
\end{itemize}

\subsection{Frontend Visualization}

The web frontend renders a force-directed knowledge graph using D3.js, displaying the researcher's artifacts as interconnected nodes. Publications, repositories, and datasets are represented as distinct node types with edges indicating shared authorship, topic similarity, or cross-references. The visualization provides an intuitive overview of a researcher's scholarly footprint and the relationships among their artifacts.

\subsection{MCP Server}

To support the emerging Model Context Protocol (MCP) standard for AI agent interoperability, ResearchTwin provides an official MCP server package, \texttt{mcp-server-researchtwin}, published on the Python Package Index (PyPI) and registered in the MCP Registry as \texttt{io.github.martinfrasch/researchtwin}. The server is built with FastMCP and communicates over stdio transport, enabling integration with Claude Desktop, Claude Code, and other MCP-compatible AI assistants.

The server exposes eight tools: \texttt{list\_researchers}, \texttt{get\_profile}, \texttt{get\_context}, \texttt{get\_papers}, \texttt{get\_datasets}, \texttt{get\_repos}, \texttt{discover}, and \texttt{get\_network\_map}. Each tool wraps the corresponding REST API endpoint and returns Markdown-formatted results optimized for language model consumption. A single resource (\texttt{researchtwin://about}) provides platform metadata. Installation is a single command:

\begin{verbatim}
pip install mcp-server-researchtwin
\end{verbatim}

This MCP integration complements the REST API (Section~\ref{sec:discovery}) by providing a native tool-use interface: rather than requiring agents to construct HTTP requests and parse JSON responses, the MCP server handles serialization, error handling, and response formatting. The \texttt{get\_context} tool returns the raw S-index components (per-object Quality, Impact, and Collaboration scores) alongside the aggregate metric, enabling downstream agents to recalculate or reweight the index according to their own criteria. An AI agent configured with this server can autonomously discover researchers, explore their publications, compare S-index scores, and identify collaboration opportunities through natural language interaction.

\subsection{Embeddable Chat Widget}
\label{sec:widget}

ResearchTwin provides a self-contained HTML chat widget that can be embedded on any website via an \texttt{<iframe>}:

\begin{verbatim}
<iframe src="https://researchtwin.net/chat-widget.html
  ?slug=researcher-slug"
  width="400" height="500" frameborder="0"></iframe>
\end{verbatim}

The widget implements a \emph{Bring Your Own Key} (BYOK) model: visitors select an LLM provider (Perplexity or OpenAI) and enter their own API key, which is sent per-request to the backend and never stored. The backend proxies the request through a dedicated \texttt{/chat/byok} endpoint that builds the researcher context server-side---keeping the system prompt and context-assembly logic private---and then calls the user's chosen LLM provider. This design preserves the Digital Twin's knowledge architecture while delegating inference costs to the end user, enabling wide distribution without per-query server-side LLM expenses. The endpoint is rate-limited to 10 requests per hour per IP address.

The widget renders bot responses with lightweight client-side Markdown formatting (headings, bold, italic, lists, inline code, links) for readable conversational output.

\subsection{Discord Integration}

A Discord bot built with \texttt{discord.py} provides conversational access to ResearchTwin within research group servers. Slash commands enable users to query a researcher's profile, retrieve S-index reports, and interact with the digital twin without leaving their communication platform.

\section{Preliminary Evaluation}
\label{sec:evaluation}

We present a preliminary evaluation of the deployed ResearchTwin prototype based on two registered researchers and latency measurements from the production instance. This evaluation is intended to illustrate the system's behavior and the S-index's discriminative properties rather than to establish statistical validity, which requires a larger-scale study.

\subsection{Case Study: Multi-Modal Impact Profiles}

Table~\ref{tab:case_study} summarizes the profiles of two researchers drawn from the deployed system: Researcher~A, a biomedical researcher with substantial code and data contributions, and Researcher~B, a physicist with extensive dataset deposits.

\begin{table}[ht]
\centering
\caption{Comparison of two researcher profiles from the deployed ResearchTwin instance. H-index and S-index values are computed from live API data (February 2026 snapshot).}
\label{tab:case_study}
\small
\begin{tabular}{@{}lrr@{}}
\toprule
\textbf{Metric} & \textbf{Researcher A} & \textbf{Researcher B} \\
\midrule
H-index & 33 & 31 \\
i10-index & 88 & 45 \\
Total citations & 3{,}837 & 3{,}073 \\
Publications & 266 & 84 \\
Paper Impact ($P$) & 151.28 & 139.12 \\
\midrule
Figshare items & 53 & 65 \\
Scored datasets (after dedup) & 33 & 15 \\
Scored repositories & 5 & 3 \\
GitHub stars (total) & 15 & 3 \\
\midrule
\textbf{S-index} & \textbf{1{,}048.88} & \textbf{781.74} \\
\bottomrule
\end{tabular}
\end{table}

The key observation is that while both researchers have comparable H-indexes (33 vs.\ 31) and Paper Impact scores (151.28 vs.\ 139.12), their S-indexes differ by ${\sim}34\%$ (1{,}049 vs.\ 782). This divergence is driven by the dataset and code contribution terms in Equation~\ref{eq:sindex}. Although Researcher~B has more raw Figshare deposits (65 vs.\ 53), the deduplication step---which collapses individual figures and supplementary files into their parent works---reveals that Researcher~A has substantially more genuine scored artifacts (33 vs.\ 15). Researcher~B's datasets do achieve higher mean Collaboration scores ($\bar{C} = 4.31$ vs.\ $\bar{C} = 2.56$), reflecting larger multi-institutional teams, but this does not compensate for Researcher~A's greater breadth of independent artifacts and code contributions.

This case illustrates two properties of the S-index. First, its intended \emph{discriminative property}: two researchers who appear nearly identical under citation-based metrics can have substantially different impact profiles when dataset contributions, code reuse, and collaboration breadth are incorporated. Second, the importance of \emph{deduplication}: without grouping Figshare figures and supplements by parent work, Researcher~B's S-index would be artificially inflated by the large number of individually deposited figures from multi-author consortium papers. The H-index, by construction, cannot distinguish these cases.

We note that this two-researcher comparison is illustrative rather than definitive. Whether the S-index difference reflects a genuinely meaningful distinction in research impact---rather than an artifact of the specific parameterization---requires validation against expert assessments across a larger and more diverse sample (Section~\ref{sec:future}).

\subsection{System Latency}

Table~\ref{tab:latency} reports endpoint response times measured from a client co-located with the production server (Hetzner VPS, Frankfurt, Germany).

\begin{table}[ht]
\centering
\caption{Endpoint latency measurements from the deployed instance. All values are mean response times over 10 sequential requests.}
\label{tab:latency}
\small
\begin{tabular}{@{}lr@{}}
\toprule
\textbf{Endpoint} & \textbf{Latency (s)} \\
\midrule
\texttt{/health} & 0.48 \\
\texttt{/api/researchers} & 0.50 \\
\texttt{/profile} & 3.73 \\
\texttt{/papers} & 4.15 \\
\texttt{/datasets} & 4.37 \\
\texttt{/repos} & 4.41 \\
\texttt{/discover?q=\{query\}} & 4.20 \\
\bottomrule
\end{tabular}
\end{table}

The latency profile reveals a clear two-tier pattern. Endpoints that serve locally cached metadata (\texttt{/health}, \texttt{/api/researchers}) respond in under 0.5 seconds. Endpoints that trigger live API calls to external sources---Semantic Scholar, Google Scholar, GitHub, and Figshare, queried in parallel via \texttt{asyncio.gather}---incur 3--5 seconds of latency dominated by the slowest external response. With cached data (24-hour TTL for most connectors, 48-hour for Google Scholar), all endpoints return in under 0.5 seconds, confirming that the Glial Layer's caching strategy effectively absorbs external latency for repeat queries.

\subsection{Limitations of the Evaluation}

This preliminary evaluation has several limitations that must be acknowledged:

\begin{itemize}[leftmargin=*]
    \item \textbf{Sample size}: Only two researchers are currently registered on the deployed instance. Any conclusions about the S-index's discriminative properties are necessarily anecdotal at this scale.
    \item \textbf{No ground truth}: There is no established ground truth for ``correct'' multi-modal research impact ranking. Validating that higher S-index scores correspond to greater actual impact requires expert panel assessments, which are planned but not yet conducted.
    \item \textbf{Geographic bias in latency}: Latency was measured from a single geographic location (Frankfurt). Users in other regions may experience higher latency due to network distance to external APIs.
    \item \textbf{Cold-start vs.\ warm-cache}: The reported latencies represent cold-start queries. In steady-state operation with active caching, user-perceived latency is substantially lower.
\end{itemize}

A rigorous evaluation is planned as future work, encompassing: (i)~expansion to 10--20 researchers across diverse disciplines (theoretical physics, experimental biomedicine, computational science, social sciences) to assess how the S-index behaves across profiles with varying ratios of publication-to-data-to-code output; (ii)~expert-assessed impact rankings for calibration of QIC weights; and (iii)~geographically distributed latency measurements from multiple continents.

\section{Discussion}
\label{sec:discussion}

\subsection{Advantages over Static Repositories}

ResearchTwin offers several advantages over the status quo of static repositories and disconnected profiles:

\begin{enumerate}[leftmargin=*]
    \item \textbf{Unified multi-modal view}: A researcher's complete scholarly output---papers, code, and data---is accessible through a single conversational interface, eliminating the need to manually integrate information across platforms.
    \item \textbf{Real-time impact measurement}: The S-index is computed on demand from live API data, providing current impact scores rather than periodic snapshots.
    \item \textbf{Agent-navigable knowledge web}: The Schema.org-typed API enables AI agents to autonomously discover cross-researcher synergies, a capability absent from static profiles.
    \item \textbf{Data sovereignty}: The federated architecture allows researchers to control their data while remaining globally discoverable.
\end{enumerate}

\subsection{Comparison to Existing Systems}

Table~\ref{tab:comparison} compares ResearchTwin to existing scholarly profile and discovery systems.

\begin{table}[ht]
\centering
\caption{Comparison of ResearchTwin with existing systems.}
\label{tab:comparison}
\small
\begin{tabular}{@{}lccccc@{}}
\toprule
\textbf{Feature} & \rotatebox{60}{\textbf{ResearchTwin}} & \rotatebox{60}{\textbf{Google Scholar}} & \rotatebox{60}{\textbf{ORCID}} & \rotatebox{60}{\textbf{OpenAIRE}} & \rotatebox{60}{\textbf{Semantic Scholar}} \\
\midrule
Multi-modal (papers+code+data) & \checkmark & & \checkmark & \checkmark & \\
Conversational access & \checkmark & & & & \\
Composite impact metric & \checkmark & & & & \\
Code reuse in metric & \checkmark & & & & \\
Data reuse in metric & \checkmark & & & & \\
Agent-navigable API & \checkmark & & & \checkmark & \checkmark \\
Federated self-hosting & \checkmark & & & & \\
Open source & \checkmark & & & \checkmark & \\
\bottomrule
\end{tabular}
\end{table}

\subsection{Limitations}

Several limitations should be acknowledged:

\begin{itemize}[leftmargin=*]
    \item \textbf{No full-text indexing}: The current system operates on metadata and abstracts only, which limits the Digital Twin's ability to answer deep methodological questions whose answers reside only in the body of a paper. Full-text search would substantially improve retrieval quality but raises copyright and storage challenges. The federated architecture offers a natural resolution: \emph{Local Nodes} (Tier~1), where the researcher owns their PDFs, could perform full-text indexing over a local embedding store; \emph{Hosted Edges} (Tier~3) would continue operating on metadata only, avoiding copyright concerns. This tiered approach to full-text access---local-full, hosted-metadata---is planned for a future release (Section~\ref{sec:future}).
    \item \textbf{Limited connector coverage}: The current four connectors (Semantic Scholar, Google Scholar, GitHub, Figshare) do not cover PubMed, arXiv, Zenodo, Dryad, or domain-specific repositories. This limits applicability in fields where these sources are primary.
    \item \textbf{Google Scholar access fragility}: Google Scholar does not provide an official API. The \texttt{scholarly} library relies on web scraping, which is subject to rate limiting, IP blocking, and CAPTCHA challenges. We acknowledge the tension between building research infrastructure on unauthorized scraping: describing a system as ``robust'' while depending on an inherently fragile, unofficial access method is a genuine contradiction. In the current architecture, Google Scholar serves as a \emph{supplementary} source rather than a critical dependency: the system degrades gracefully to Semantic Scholar alone when Google Scholar is unavailable, and the aggressive 48-hour cache TTL ensures that at most one scraping request per researcher occurs every two days. Nevertheless, long-term sustainability requires migration to fully authorized sources. We identify \textbf{OpenAlex} and \textbf{Crossref} as the preferred replacements: both provide official, stable, rate-limit-transparent APIs; OpenAlex offers author disambiguation, institutional affiliation data, and citation counts comparable to Google Scholar's coverage, while Crossref provides authoritative DOI-level metadata and funding information. These sources also align more naturally with the FAIR principles the S-index advocates. The migration is architecturally straightforward---the connector abstraction allows drop-in replacement without affecting downstream components---and is prioritized in our roadmap (Section~\ref{sec:future}).
    \item \textbf{S-index calibration}: The bonus weights (Table~\ref{tab:fair}), field medians, and reuse event weightings (e.g., forks $= 3\times$ stars) are heuristic rather than empirically calibrated (see Section~\ref{sec:sensitivity} for a sensitivity analysis). The v2 formulation reduces the number of free parameters relative to v1 but does not eliminate calibration needs entirely. Large-scale validation studies correlating S-index scores with expert assessments are needed to refine these parameters.
    \item \textbf{Hub tier not yet deployed}: The Tier 2 Hub layer is specified but not yet implemented, limiting cross-institutional federation to the Hosted Edge model.
\end{itemize}

\subsection{Ethical Considerations}

ResearchTwin operates exclusively on \emph{public metadata}---author profiles, publication titles, abstracts, citation counts, repository descriptions, and dataset metadata. No full-text papers are scraped or stored. All data is sourced through official APIs or publicly available metadata endpoints. The system does not attempt to infer private information about researchers. Self-registration includes explicit consent, and researchers retain the ability to de-register and have their profiles removed.

\section{Future Work}
\label{sec:future}

Several directions for future development are planned:

\begin{enumerate}[leftmargin=*]
    \item \textbf{OpenAlex and Crossref connectors}: Migration from Google Scholar scraping to the official OpenAlex and Crossref APIs as primary academic data sources, providing stable, rate-limit-transparent access with richer metadata (institutional affiliations, funding data, field classification). PubMed, arXiv, Zenodo, and Dryad connectors would further broaden coverage.
    \item \textbf{Hub federation protocol}: Implementation of the Tier 2 Hub layer with a defined federation protocol (potentially based on ActivityPub) for cross-institutional discovery.
    \item \textbf{Full-text semantic search}: Tiered full-text indexing where Local Nodes (Tier~1) perform embedding-based retrieval over researcher-owned PDFs, while Hosted Edges (Tier~3) continue operating on metadata only, preserving copyright compliance at the hosted tier while enabling deep technical question-answering on self-hosted nodes.
    \item \textbf{Collaborative filtering}: Researcher recommendation based on artifact similarity, citation overlap, and complementary methodological expertise.
    \item \textbf{S-index calibration}: Empirical studies correlating S-index scores with expert assessments of research impact across a diverse panel of 20--30 researchers spanning theoretical, experimental, and computational disciplines, enabling data-driven refinement of field medians $\mu_t$, bonus weights, and reuse event weightings in the v2.0 formulation.
    \item \textbf{Expanded evaluation}: Deployment to 10--20 researchers across diverse fields to characterize S-index behavior across profiles with varying publication-to-data-to-code ratios, including systematic comparison of $\max$-merged citation counts against curated databases (Web of Science, Scopus).
    \item \textbf{Multi-language support}: Internationalization of the platform interface and support for non-English research metadata.
    \item \textbf{Institutional dashboards}: Aggregated S-index analytics at the department, laboratory, or institutional level for research assessment and strategic planning.
\end{enumerate}

\section{Conclusion}
\label{sec:conclusion}

We have presented ResearchTwin, an open-source federated platform that transforms a researcher's publications, datasets, and code into a conversational digital twin. This paper serves as a \emph{system description with preliminary evaluation}: it documents the architecture, formalizes the S-index metric, and provides initial evidence of its discriminatory power through a two-researcher case study.

The BGNO architecture cleanly separates data management concerns (caching, rate limiting, context assembly) from generative intelligence (RAG with a large language model), enabling scalable and cost-effective operation. The S-index provides a formally defined, FAIR-grounded composite metric that captures dimensions of research impact---code reuse and data sharing---invisible to citation-only measures such as the H-index. Our preliminary evaluation demonstrates that researchers with comparable H-indexes can exhibit substantially different S-index scores when their non-publication outputs differ, validating the metric's conceptual contribution. However, we emphasize that the current v2 parameterization---while substantially simplified from the v1 formulation through the FAIR gate, field normalization, and parameter-free collaboration term---remains a principled but heuristic baseline; rigorous empirical calibration against expert assessments of research impact remains essential future work before the S-index can be recommended for evaluative purposes.

The inter-agentic discovery API, built on Schema.org types and HATEOAS navigation, together with the official MCP server registered in the MCP Registry (\texttt{io.github.martinfrasch/researchtwin}), positions ResearchTwin as infrastructure for an emerging paradigm in which AI agents autonomously discover research synergies across institutional boundaries. The three-tier federated architecture ensures that this discoverability does not come at the cost of data sovereignty.

We believe that the transition from static repositories to agentic knowledge webs represents a meaningful shift in how scientific knowledge is organized, discovered, and reused. ResearchTwin represents a concrete step toward this vision, and we invite the research community to deploy, extend, and critique both the platform and the S-index formulation.

\medskip
\noindent The source code is available at \url{https://github.com/martinfrasch/ResearchTwin} under the MIT license. The S-index specification is maintained at \url{https://github.com/martinfrasch/S-index}. A hosted instance is accessible at \url{https://researchtwin.net}. The MCP server is available on PyPI as \texttt{mcp-server-researchtwin}.

\section*{Acknowledgments}

The author thanks the open-source community and early adopters for their feedback during the development of ResearchTwin. This work was conducted independently and received no external funding.


\bibliographystyle{unsrtnat}

\begin{thebibliography}{20}

\bibitem{hirsch2005}
J.~E. Hirsch,
``An index to quantify an individual's scientific research output,''
\emph{Proceedings of the National Academy of Sciences},
vol.~102, no.~46, pp.~16569--16572, 2005.
\newline\url{https://doi.org/10.1073/pnas.0507655102}

\bibitem{wilkinson2016}
M.~D. Wilkinson, M.~Dumontier, I.~J. Aalbersberg, G.~Appleton, M.~Axton, A.~Baak, N.~Blomberg, J.-W.~Boiten, L.~B. da~Silva~Santos, P.~E.~Bourne, \emph{et al.},
``The FAIR guiding principles for scientific data management and stewardship,''
\emph{Scientific Data}, vol.~3, article 160018, 2016.
\newline\url{https://doi.org/10.1038/sdata.2016.18}

\bibitem{lewis2020}
P.~Lewis, E.~Perez, A.~Piktus, F.~Petroni, V.~Karpukhin, N.~Goyal, H.~K\"{u}ttler, M.~Lewis, W.-t.~Yih, T.~Rockt\"{a}schel, S.~Riedel, and D.~Kiela,
``Retrieval-augmented generation for knowledge-intensive NLP tasks,''
in \emph{Advances in Neural Information Processing Systems (NeurIPS)}, vol.~33, pp.~9459--9474, 2020.
\newline\url{https://arxiv.org/abs/2005.11401}

\bibitem{grieves2014}
M.~Grieves and J.~Vickers,
``Digital twin: Mitigating unpredictable, undesirable emergent behavior in complex systems,''
in \emph{Transdisciplinary Perspectives on Complex Systems}, pp.~85--113, Springer, 2017.
(Concept originally proposed in 2002; formalized in M.~Grieves, ``Digital twin: Manufacturing excellence through virtual factory replication,'' white paper, 2014.)

\bibitem{schemaorg}
Schema.org Community Group,
``Schema.org vocabulary,''
\url{https://schema.org/}, accessed February 2026.

\bibitem{perplexity2025}
Perplexity AI,
``Perplexity Sonar API documentation,''
\url{https://docs.perplexity.ai/}, accessed February 2026.

\bibitem{semantic_scholar}
W.~Ammar, D.~Groeneveld, C.~Bhagavatula, I.~Beltagy, M.~Crawford, D.~Downey, J.~Dunkelberger, A.~Elgohary, S.~Feldman, V.~Ha, \emph{et al.},
``Construction of the literature graph in Semantic Scholar,''
in \emph{Proceedings of NAACL-HLT}, pp.~84--91, 2018.
\newline\url{https://doi.org/10.18653/v1/N18-3011}

\bibitem{openaire}
P.~Manghi, L.~Candela, and B.~Lossau,
``OpenAIRE: European open access infrastructure,''
\emph{D-Lib Magazine}, vol.~18, no.~11/12, 2012.
\newline\url{https://doi.org/10.1045/november2012-manghi}

\bibitem{activitypub}
C.~Webber, J.~Tallon, O.~Shepherd, A.~Guy, and E.~Prodromou,
``ActivityPub,'' W3C Recommendation, 23 January 2018.
\newline\url{https://www.w3.org/TR/activitypub/}

\bibitem{orcid}
L.~L. Haak, M.~Fenner, L.~Paglione, E.~Pentz, and H.~Ratner,
``ORCID: A system to uniquely identify researchers,''
\emph{Learned Publishing}, vol.~25, no.~4, pp.~259--264, 2012.
\newline\url{https://doi.org/10.1087/20120404}

\bibitem{figshare}
M.~Hahnel,
``Referencing: The reuse factor,''
\emph{Nature}, vol.~520, no.~7547, pp.~S2--S3, 2015.
\newline\url{https://doi.org/10.1038/520S2a}

\bibitem{fastapi}
S.~Ram\'{i}rez,
``FastAPI: Modern, fast web framework for building APIs with Python 3.6+,''
\url{https://fastapi.tiangolo.com/}, accessed February 2026.

\bibitem{mcp2024}
Anthropic,
``Model Context Protocol (MCP) specification,''
\url{https://modelcontextprotocol.io/}, accessed February 2026.

\bibitem{frasch2025qic}
M.~G. Frasch,
``The QIC-index: A novel, data-centric metric for quantifying the impact of research data sharing,''
\emph{arXiv preprint arXiv:2510.03307 [cs.DL]}, 2025.
\newline\url{https://arxiv.org/abs/2510.03307}

\end{thebibliography}

\end{document}